\documentclass[preprint]{vldb}
\usepackage[utf8]{inputenc}
\usepackage{ifthen}
\usepackage{comment}
\usepackage{xcolor}
\usepackage{amsmath}
\usepackage{url}

\usepackage{fancyvrb}

\newboolean{sigmod}
\setboolean{sigmod}{true}

\excludecomment{journal}
\includecomment{sigmod}
\newcommand{\journo}{}

\newcommand{\Calvin}{\textsc{Calvin}}
\newcommand{\Cassandra}{\textsc{Cassandra}}
\newcommand{\Dgraph}{\textsc{Dgraph}}
\newcommand{\Jepsen}{\textsc{Jepsen}}
\newcommand{\Knossos}{\textsc{Knossos}}
\newcommand{\Gecode}{\textsc{Gecode}}

\newcommand{\Name}{\textsc{Elle}}
\newcommand{\LineUp}{\textsc{Line-Up}}
\newcommand{\Porcupine}{\textsc{Porcupine}}
\newcommand{\TiDB}{\textsc{TiDB}}
\newcommand{\YugaByte}{\textsc{YugaByte DB}}
\newcommand{\FaunaDB}{\textsc{FaunaDB}}
\newcommand{\Spanner}{\textsc{Spanner}}
\newcommand{\Percolator}{\textsc{Percolator}}
\newcommand{\Postgres}{\textsc{Postgres}}
\newcommand{\QuickCheck}{\textsc{QuickCheck}}

\newtheorem{axiom}{Axiom}

\begin{document}

\title{Elle: Inferring Isolation Anomalies from Experimental Observations}

\numberofauthors{2}
\author{
\alignauthor
    Kyle Kingsbury\\
    \affaddr{Jepsen}\\
     \email{aphyr@jepsen.io}
\alignauthor
    Peter Alvaro\\
    \affaddr{UC Santa Cruz}\\
    \email{palvaro@ucsc.edu}
}


\date{June 2019}

\toappear{}
\maketitle

\begin{abstract}
Users who care about their data store it in databases, which (at least in principle) guarantee some form of transactional isolation. However, experience shows \cite{hermitage,jepsen-reports} that many databases do not provide the isolation guarantees they claim. With the recent proliferation of new distributed databases, demand has grown for checkers that can, by generating client workloads and injecting faults, produce anomalies that witness a violation of a stated guarantee. An ideal checker would be sound (no false positives), efficient (polynomial in history length and concurrency), effective (finding violations in real databases), general (analyzing many patterns of transactions), and informative (justifying the presence of an anomaly with understandable counterexamples). Sadly, we are aware of no checkers that satisfy these goals.

We present \Name{}: a novel checker which infers an Adya-style dependency graph between client-observed transactions. It does so by carefully selecting database objects and operations when generating histories, so as to ensure that the results of database reads reveal information about their version history. \Name{} can detect every anomaly in Adya et al's formalism \cite{adya-gil} (except for predicates), discriminate between them, and provide concise explanations of each. 

This paper makes the following contributions: we present \Name{}, demonstrate its soundness, measure its efficiency against the current state of the art, and give evidence of its effectiveness via a case study of four real databases.
\end{abstract}

\section{Introduction}
\label{sec:intro}

Database systems often offer multi-object transactions at varying isolation levels, such as serializable or read committed. However, design flaws or bugs in those databases may result in weaker isolation levels than claimed. In order to verify whether a given database actually provides claimed safety properties, we can execute transactions against the database, record a concurrent history of how those transactions completed, and analyze that history to identify invariant violations. This property-based approach to verification is especially powerful when combined with fault injection techniques.~\cite{random-testing}

Many checkers use a particular pattern of transactions, and check that under the expected isolation level, some hand-proved invariant(s) hold. For instance, one could check for long fork (an anomaly present in parallel snapshot isolation) by inserting two records $x$ and $y$ in separate transactions, and in two more transactions, reading both records. If one read observes $x$ but not $y$, and the other observes $y$ but not $x$, then we have an example of a long fork, and can conclude that the system does not provide snapshot isolation---or any stronger isolation level.

These checkers are generally efficient (i.e. completing in polynomial time), and do identify bugs in real systems, but they have several drawbacks. They find a small number of anomalies in a specific pattern of transactions, and tell us nothing about the behavior of other patterns. They require hand-proven invariants: one must show that for chosen transactions under a given consistency model, those invariants hold. They also do not compose: transactions we execute for one checker are, in general, incompatible with another checker. Each property may require a separate test.

Some researchers have designed more general checkers which can analyze broader sets of possible transactions. For instance, \Knossos~\cite{knossos} and \Porcupine~\cite{porcupine} can verify whether an arbitrary history of operations over a user-defined datatype is linearizable, using techniques from Wing \& Gong~\cite{wing-concurrent} and Lowe~\cite{lowe}. Since strict-1SR is equivalent to linearizability (where operations are transactions, and the linearizable object is a map), these checkers can be applied to strict serializable databases as well. While this approach does find anomalies in real databases, its use is limited by the NP-complete nature of linearizability checking, and the combinatorial explosion of states in a concurrent multi-register system.

Serializability checking is \textit{also} (in general) NP-complete~\cite{papadimitriou}---and unlike linearizability, one cannot use real-time constraints to reduce the search space. Following the abstract execution formalism of Cerone, Bernardi, and Gotsman~\cite{cerone}, Kingsbury attempted to verify serializability by identifying write-read dependencies between transactions, translating those dependencies to an integer constraint problem on transaction orders~\cite{gretchen}, and applying off-the-shelf constraint solvers like \Gecode{}~\cite{gecode} to solve for a particular order. This approach works, but, like \Knossos, is limited by the NP-complete nature of constraint solving. Histories of more than a hundred-odd transactions quickly become intractable. Moreover, constraint solvers give us limited insight into \textit{why} a particular transaction order was unsolvable. They can only tell us whether a history is serializable or not, without insight into specific transactions that may have gone wrong. Finally, this approach cannot distinguish between weaker isolation levels, such as snapshot isolation or read committed.

What would an ideal checker for transaction isolation look like? Such a checker would accept many patterns of transactions, rather than specific, hand-coded examples. It would distinguish between different types of anomalies, allowing us to verify stronger (e.g. strict-1SR) and weaker (e.g. read uncommitted) isolation levels. It ought to localize the faults it finds to specific subsets of transactions. Of course, it should do all of this efficiently.

In this paper, we present \Name{}: an isolation checker for black-box databases. Instead of solving for a transaction order, \Name{} uses its knowledge of the transactions issued by the client, the objects written, and the values returned by reads to reason about the possible dependency graphs of the opaque database in the language of Adya's formalism.~\cite{adya-gil} While \Name{} can make limited inferences from read-write registers, it shines with richer datatypes, like append-only lists.

All history checkers depend on the system which \textit{generated} their transactions. \Name{}'s most powerful analysis requires that we generate histories in which reads of an object correspond return its entire version history, and where a unique mapping exists between versions and transactions. However, we show that generating histories which allow these inferences is straightforward, that the required datatypes are broadly supported, and that these choices do not prevent \Name{} from identifying bugs in real-world database systems.

\section{The Adya Formalism}
In their 2000 paper ``Generalized Isolation Level Definitions'', Adya, Liskov, and O'Neil formalized a variety of transactional isolation levels in terms of proscribed behaviors in an abstract history $H$. Adya et al.'s histories (hereafter: "Adya histories") comprise of a set of transactions, an \textit{event order} which encodes the order of operations in those transactions, and a \textit{version order} $\ll$: a total order over installed versions of each object. Two anomalies, G1a (aborted reads) and G1b (intermediate reads), are defined as transactions which read versions written by aborted transactions, or which read versions from the middle of some other transaction, respectively. The remainder are defined in terms of cycles over a \textit{Direct Serialization Graph} (DSG), which captures the dependencies between transactions. Setting aside predicates, Adya et al.'s dependencies are:

\begin{itemize}
    \item \textbf{Directly write-depends}. $T_i$ installs $x_i$, and $T_j$ installs $x$'s next version
    \item \textbf{Directly read-depends}. $T_i$ installs $x_i$, $T_j$ reads $x_i$
    \item \textbf{Directly anti-depends}. $T_i$ reads $x_i$, and $T_j$ installs $x$'s next version
\end{itemize}

The direct serialization graph $DSG(H)$ is a graph over transactions in some history $H$, whose edges are given by these dependencies. A G0 anomaly is a cycle in the DSG comprised entirely of write dependencies. G1c anomalies include read dependencies. Instances of G2 involve at least one anti-dependency (those with exactly one are G-single).

This is a tantalizing model for several reasons. Its definitions are relatively straightforward. Its anomalies are \textit{local}, in the sense that they involve a specific set of transactions. We can point to those transactions and say, "Things went wrong here!"---which aids debugging. Moreover, we can check these properties in \textit{linear time}: intermediate and aborted reads are straightforward to detect, and once we've constructed the dependency graph, cycle detection is solvable in O(vertices + edges) time, thanks to Tarjan’s algorithm for strongly connected components~\cite{tarjan}.

However, there is one significant obstacle to working with an Adya history: \textit{we don't have it}. In fact, one may not even exist. The database system may not have any concept of a version order, or it might not expose that ordering information to clients.

As an example, consider Adya et al.'s history $H_{serial}$, as it would be observed by clients. Is this history serializable?

\begin{itemize}
    \item [$T_1$:] $w(z_1), w(x_1), w(y_1), c$
    \item [$T_2$:] $r(x_1), w(y_2), c$
    \item [$T_3$:] $w(x_3), r(y_2), w(z_3), c$
\end{itemize}

$T_2$ read $x_1$, so it must read-depend on $T_1$, and likewise, $T_3$ must read-depend on $T_2$. What about $T_1$ and $T_2$’s writes to $y$? Which overwrote the other? As Crooks et al observe [17], we cannot tell, because we lack a key component in Adya’s formalism: the version order. $H_{serial}$ includes additional information ($x_1 \ll x_3, y1 \ll y_2, z_1 \ll z_3$) which is invisible to clients. We might consider deducing the version order from the real-time order in which writes or commits take place, but Adya et al explicitly rule this out, since optimistic and multi-version implementations might require the freedom to commit earlier versions later in time. Moreover, network latency may make it impossible to precisely determine concurrency windows.

We would like to be able to \textit{infer} an Adya history based purely on the information available to clients of the system, which we call an \textit{observation}. When a client submits an operation to a database, it generally knows what kind of operation it performed (e.g. a read, a write, etc.), the object that operation applies to, and the arguments it provided. For instance, a client might write the value 5 to object $x$. If the database returns a response for an operation, we may also know a return value; for instance, that a read of $x$ returned the number 5.

Clients know the order of operations within every transaction. They may also know whether a transaction was definitely aborted, definitely committed, or could be either, based on whether a commit request was sent, and how (or if) the database acknowledged that request. Clients can also record their own per-client and real-time orders of events.

This is not, in general, enough information to go on. Consider a pair of transactions which set $x$ to 1 and 2, respectively. In the version order, does $1$ or $2$ come first? We can't say! Or consider an indeterminate transaction whose effects are not observed. Did it commit? We have no way to tell.

This implies there might be many possible histories which are compatible with a given observation. Are there conditions under which only one history is possible? Or, if more than one is possible, can we infer something about the structure of \textit{all} of them which allows us to identify anomalies?

We argue that yes: one \textit{can} infer properties of every history compatible with a given observation, by taking advantage of datatypes which allow us to trace the sequence of versions which gave rise to a particular version of an object, and which let us recover the particular writes which gave rise to those versions. Next, we provide an intuition for how this can be accomplished.

\section{Deducing Dependencies}
\begin{figure*}[htb!]
\begin{center}
\begin{tabular}{l l l l}
 Object & Versions & $x_{init}$ & Writes \\
\hline
 Register   & Any       & nil  & $w(x_i, a) \xrightarrow{} (a, nil)$ \\
 Counter    & Integers  & 0    & $w(x_i, a) \xrightarrow{} (x_i+a, nil)$ \\
 Set Add    & Sets      & \{\} & $w(x_i, a) \xrightarrow{} (x_i \cup \{a\}, nil)$ \\
 List Append & Lists    & []   & $w([e_1, \ldots, e_n], a) \xrightarrow{} ([e_1, \ldots, e_n, a], nil)$ \\
\end{tabular}
\end{center}
\caption{Example objects.}
\label{fig:example-objects}
\end{figure*}

Consider a set of observed transactions interacting with some read-write \textbf{register} $x$. One transaction $T_j$ read $x$ and observed some version $x_i$. Another transaction $T_i$ wrote $x_i$ to $x$. In general, we cannot tell whether $T_i$ was the transaction which produced $x_i$, because some \textit{other} transaction might have written $x_i$ as well. However, if we know that no other transaction wrote $x_i$, then we can recover the particular transaction which wrote $x_i$: $T_i$. This allows us to infer a direct write-read dependency: $T_i <_{wr} T_j$.

If every value written to a given register is unique\footnotemark, then we can recover the transaction which gave rise to any observed version. We call this property \textit{recoverability}: every version we observe can be mapped to a specific write in some observed transaction.

\footnotetext{This approach is used by Crooks et al, and has a long history in the literature.}

Recoverability allows us to infer read dependencies. However, inferring write- and anti-dependencies takes more than recoverability: we need the \textit{version order} $\ll$. Read-write registers make inferring $\ll$ impossible in general: if two transactions set $x$ to $x_i$ and $x_j$ respectively, we can't tell which came first.

In a sense, blind writes to a register ``destroy history''. If we used a compare-and-set operation, we could tell something about the preceding version, but a blind write can succeed regardless of what value was present before. Moreover, the version resulting from a blind write carries no information about the previous version. Let us therefore consider richer datatypes, whose writes \textit{do} preserve some information about previous versions.

For instance, we could take increment operations on \textbf{counters}, starting at 0. If every write is an increment, then the version history of a register should be something like (0, 1, 2, \ldots). Given two transactions $T_i$ and $T_j$, both of which read object $x$, we can construct a read-read dependency $T_i <_{rr} T_j$ if $T_i$'s observed value of $x$ is smaller than $T_j$'s. However, any non-trivial increment history is \textit{non-recoverable}, because we can't tell \textit{which} increment produced a particular version. This keeps us from inferring write-write, write-read, and read-write dependencies. We could return the resulting version from our writes, but this works only when the client receives acknowledgement of its request(s).

What if we let our objects be \textbf{sets} of elements, and had each write add a unique element to a given set? Like counters, this lets us recover read-read dependencies whenever one version is a proper superset of another. Moreover, we can recover some (but not all) write-write, write-read, and read-write dependencies. Consider these observed transactions:

\begin{itemize}
    \item[$T_0$:] $read(x, \{0\})$
    \item[$T_1$:] $add(x, 1)$
    \item[$T_2$:] $add(x, 2)$
    \item[$T_3$:] $read(x, \{0, 1, 2\})$
\end{itemize}

Since $\{0, 1, 2\}$ is a proper superset of $\{0\}$, we know that $T_3$ read a higher version of $x$ than $T_0$, and can infer $T_0 <_{rr} T_3$. In addition, we can infer that $T_1 <_{wr} T_3$ and $T_2 <_{wr} T_3$, since their respective elements 1 and 2 were both visible to $T_3$. Conversely, since $T_0$'s read of {0} did \textit{not} include 1 or 2, we can infer that $T_0 <_{rw} T_1$ and $T_0 <_{rw} T_2$: anti-dependency relations! We cannot, however, identify the write-write dependency between $T_1$ and $T_2$: the database could have executed $T_1$ and $T_2$ in either order with equivalent effects, because sets are order-free. Only a read of $\{0, 1\}$ or $\{0, 2\}$ could resolve this ambiguity.

So, let's add order to our values. Let each version be an ordered \textbf{list}, and a write to $x$ append a unique value to $x$. Then any read of $x_i$ tells us the order of all versions written prior: $x_i = [1, 2, 3]$ implies that $x$ took on the versions $[], [1], [1, 2]$, and $[1, 2, 3]$ in exactly that order. We call this property \textit{traceability}. Since appends add a unique element to the \textit{end} of $x$, we can also infer the exact write which gave rise to any observed version: these versions are \textit{recoverable} as well.

As with counters and sets, we can use traceability to reconstruct read-read, write-read, and read-write dependencies---but because we have the full version history, we can precisely identify read-write and write-write dependencies for every transaction whose writes were observed by some read. We know, for instance, that the transaction which added 2 to $x$ must write-depend on the transaction which added 1 to $x$, because we have a read of $x$ wherein 1 immediately precedes 2. There may be some writes near the end of a history which are never observed, but so long as histories are long and include reads every so often, the unknown fraction of a version order can be made relatively small.

Recoverability and traceability are the key to inferring dependencies between observed transactions, but we have glossed over the mapping between observed dependencies and Adya histories, as well as the challenges arising from aborted and indeterminate transactions. In the following section, we discuss these issues more rigorously.



\section{Formal Model}
\label{sec:formalism}

In this section we present our abstract model of databases and transactional isolation.  We establish the notions of \emph{traceability} and \emph{recoverability}, which we prove to be sufficient to reason directly from external observations to internal histories.  We show that this approach is \emph{sound}: that is, any anomalies identified in an observation must be present in any Adya history consistent with that observation.

\begin{sigmod}

Due to space considerations, we do not present the formal definitions of traceability and recoverability or their accompanying proofs here; instead, we summarize these results.

\end{sigmod}

\subsection{Preliminaries}

We begin our formalism by defining a model of a database, transactions, and histories, closely adapted from Adya et al. We omit predicates for simplicity, and generalize Adya's read-write registers to objects supporting arbitrary transformations as writes, resulting in a \textit{version graph}. We constrain Adya's \textit{version order $\ll$} to be compatible with this version graph. This generalization introduces a new class of anomaly, \textit{dirty updates}, which we define in Section~\ref{sec:dirty}.

\subsubsection{Objects, Operations, Databases}

An \textit{object} $x$ is a mutable datatype, consisting of a set of \textit{versions}\footnote{For simplicity, we assume versions \textit{are} values, and that versions do not repeat in a history.}, written $x_i$, $x_j$, etc., an \textit{initial version}, labeled $x_{init}$, and a set of \textit{object operations}.

An object operation represents a state transition between two versions $x_i$ and $x_j$ of some object $x$. Object operations take an argument $a$ and produce a return value $r$. 
\begin{sigmod}
We write this as $f(x, x_i, a) \xrightarrow{} (x_j, r)$. Where the object, argument, return value, or return tuple can be inferred from context or are unimportant, we may omit them: $f(x_i, a) \xrightarrow{} (x_j)$, $f(x_i) \xrightarrow{} (x_j)$, $f(x_i)$, etc.
\end{sigmod}
\begin{journal}
\journo{}
Formally, an operation is a 5-tuple of objects, pre-versions, arguments, post-versions, and return values, but for clarity, we write $f(x, x_i, a) \xrightarrow{} (x_j, r)$. Where the object, argument, return value, or return tuple can be inferred from context or are unimportant, we may omit them, e.g. $f(x_i, a) \xrightarrow{} (x_j)$, $f(x_i) \xrightarrow{} (x_j)$, $f(x_i)$, etc.
\end{journal}

Like Adya, we consider two types of operations: reads ($r)$ and writes ($w$). A \textit{read} takes no argument, leaves the version of the object unchanged, and returns that version: $r(x_i, nil) \xrightarrow{} (x_i, x_i)$.

As we show in Figure~\ref{fig:example-objects}, a write operation $w$ changes a version somehow. Write semantics are object-dependent. Adya's model (like much work on transactional databases) assumes objects are registers, and that writes blindly replace the current value. Our other three objects incorporate increasingly specific dependencies on previous versions. In this section we are primarily concerned with objects like list append, but we provide definitions for the first three objects as illustrative examples.

The versions and write operations on an object $x$ together yield a \textit{version graph} $v_x$: a directed graph whose vertices are versions, and whose edges are write operations.

A \textit{database} is a set of objects $x$, $y$, etc.

\subsubsection{Transactions}

A \textit{transaction} is a list (a totally ordered set) of object operations, followed by at most one \textit{commit} or \textit{abort} operation. Transactions also include a unique identifier for disambiguation.

We say a transaction is \textit{committed} if it ends in a commit, and \textit{aborted} if it ends in an abort. We call the version resulting from a transaction's last write to object $x$ its \textit{final version} of $x$; any other writes of $x$ result in \textit{intermediate versions}. If a transaction commits, we say it \textit{installs} all final versions; we call these \textit{installed versions}. Unlike Adya, we use \textit{committed versions} to refer to versions written by committed transactions, including intermediate versions. Versions from transactions which did not commit are called \textit{uncommitted versions}.

The initial version of every object is considered committed.

\subsubsection{Histories}

Per Adya et al, a \textit{history} $H$ comprises a set of transactions $T$ on objects in a database $D$, a partial order $E$ over operations in $T$, and a version order $\ll$ over versions of the objects in $D$.

The event order $E$ has the following constraints:

\begin{itemize}
    \setlength\itemsep{0em}
    \item It preserves the order of all operations within a transaction, including commit or abort events.
    \item One cannot read from the future: if a read $r(x_i)$ is in $E$, then there must exist a write $w \xrightarrow{} (x_i)$ which precedes it.
    \item Transactions observe their own writes: if a transaction $T$ contains $w(x_i)$ followed by $r(x_j)$, and there exists no $w(x_k)$ between the write and read of $x_i$ in $T$, $x_i = x_j$.
    \item The history must be complete: if $E$ contains a read or write operation for a transaction $T_i$, $E$ must contain a commit or abort event for $T_i$.
\end{itemize}

The version order $\ll$ has two constraints, per Adya et al
\begin{itemize}
    \setlength\itemsep{0em}
    \item It is a total order over installed versions of each individual object; there is no ordering of objects written by uncommitted, aborted, or intermediate transactions.
    \item The version order for each object $x$ (which we write $\ll_x$) begins with the initial version $x_{init}$.
\end{itemize}

Driven by the intuition that the version order should be consistent with the order of operations on an object, that $x_i \ll x_j$ means $x_j$ overwrote $x_i$, and that cycles in $\ll$ are meant to capture anomalies like \textit{circular information flow}, we introduce an additional constraint that was not necessary in Adya et al.'s formalism: the version order $\ll$ should be consistent with the version graphs $v_x$, $v_y$, \ldots.

\begin{sigmod}
Specifically, if $x_i \ll x_j$, there exists a path from $x_i$ to $x_j$ in $v_x$. It would be odd if one appended 3 to the list $[1, 2]$, obtaining $[1, 2, 3]$, and yet the database believed $[1, 2, 3] \ll [1, 2]$.
\end{sigmod}

\begin{journal}
\journo
\begin{axiom}
\label{version-order->version-graph}
If $x_i \ll x_j$, there exists a path from $x_i$ to $x_j$ in $v_x$.
\end{axiom}

Since Adya's formalism uses registers, their version graphs are \textit{totally connected}: any version can be replaced with any other version in a single write. Axiom \ref{version-order->version-graph} is therefore trivial for registers, and need not be stated. For more complex datatypes, we believe this constraint is warranted: it would be odd if one appended 3 to the list $[1, 2]$, obtaining $[1, 2, 3]$, and yet the database believed $[1, 2, 3] \ll [1, 2]$.

\end{journal}

\subsubsection{Dependency Graphs}

We define write-write, read-write, and write-read dependencies between transactions, adapted directly from Adya's formalism.  
\begin{journal}
\journo
We also define a fourth dependency class, read-read dependencies, that is useful when writes cannot be related to reads via external observations.
\end{journal}

Finally, we (re-)define the \textit{Direct Serialization Graph} using those dependencies. The anomalies we wish to find are expressed in terms of that serialization graph.

\begin{itemize}
    \setlength\itemsep{0em}
    \item \textbf{Direct write-write dependency}. A transaction $T_j$ \textit{directly ww-depends} on $T_i$ if $T_i$ installs a version $x_i$ of $x$, and $T_j$ installs $x$'s next version $x_j$, by $\ll$.
    \item \textbf{Direct write-read dependency}. A transaction $T_j$ \textit{directly wr-depends} on $T_i$ if $T_i$ installs some version $x_i$ and $T_j$ reads $x_i$.\footnotemark{}
    \item \textbf{Direct read-write dependency}. A transaction $T_j$ \textit{directly rw-depends} on $T_i$ if $T_i$ reads version $x_i$ of $x$, and $T_j$ installs $x$'s next version in $\ll$.

\begin{journal}
\journo{}
    \item \textbf{Direct read-read dependency}. A transaction $T_j$ \textit{directly rr-depends} on $T_i$ if $T_i$ reads version $x_i$ of $x$, and $T_j$ reads $x$'s next version in $\ll$.
\end{journal}    

\end{itemize}

\footnotetext{It appears that Adya et al's read dependencies don't rule out a transaction depending on itself. We preserve their definitions here, but assume that in serialization graphs, $T_i \neq T_j$.}

\begin{journal}
\journo{}

In Adya et al's formalism, rr-depends is redundant: any rr-dependency implies an rw-dependency followed by a wr-dependency. We mention it here because for some objects, like counters, it may be the most useful dependency we can infer.
\end{journal}

A \textit{Direct Serialization Graph}, or \textit{DSG}, is a graph of dependencies between transactions. The DSG for a history $H$ is denoted DSG($H$). If $T_j$ ww-depends on $T_i$, there exists an edge labeled \textit{ww} from $T_i$ to $T_j$ in DSG(H), and similarly for wr- and rw-dependencies. 
\begin{journal}
\journo{}
One could construct alterate serialization graphs involving other dependencies, such as read-read dependencies, process orders, or realtime orders, but we leave those as an exercise to the reader.
\end{journal}

\subsubsection{Dirty Updates}

Adya defines \textit{G1a: aborted read} as a committed transaction $T_2$ containing a read of some value $x_i$ which was written by an aborted transaction $T_1$. However, our abstract datatypes allow a transaction to commit a \textit{write} which incorporates aborted state.

We therefore define a complementary phenomenon to G1a, intended to cover scenarios in which information ``leaks'' from uncommitted versions to committed ones via writes. A history exhibits \textit{dirty update} if it contains an uncommitted transaction $T_1$ which writes $x_i$, and a committed transaction $T_2$ which contains a write acting on $x_i$.

\subsubsection{Traceable Objects, Version Orders and Version Graphs}

We now define a class of \textit{traceable objects}, which permit recovery of the version graph and object operations resulting in any particular version. 

Recall that for an object $x$, the version graph $v_x$ is comprised of versions connected by object operations. We call a path in $v_x$ from $x_{init}$ to some version $x_i$ a \textit{trace} of $x_i$.
Intuitively, a trace captures the version history of an object.

We say an object $x$ is \textit{traceable} if every version $x_i$ has exactly one trace; i.e. $v_x$ is a tree.

Given a history with version order $\ll$, we call the largest version of $x$ (by $\ll_x$) $x_{max}$. Because $\ll$ is a total order over installed versions, and because a path in the version graph exists between any two elements ordered by $\ll$, it follows that every committed version of $x$ is in the trace of $x_{max}$. Moreover, for any installed version $x_i$ of $x$, we can recover the prefix of $\ll_x$ up to and including $x_i$ simply by removing intermediate and aborted versions from the trace of $x_i$.

\begin{journal}
\journo{}

\begin{lemma}
\label{traceable->x_max-trace-all-committed}
If $x$ is a traceable object in a history $H$, every committed version of $x$ is in the trace of $x_{max}$.
\end{lemma}

This follows from axiom \ref{version-order->version-graph}, and the fact that the version order $\ll$ is total over installed versions.

\begin{lemma}
\label{x_i->ll_x-prefix}

If $x_i$ is an installed version of a traceable object $x$, removing intermediate and aborted versions from the versions in the trace of $x_i$ yields the prefix of $\ll_x$ up to and including $x_i$.
\end{lemma}

\begin{proof}
Consider the case where $x_i$ is $x_{init}$. Then the trace of $x_i$ is the singleton path $[x_{init}]$, and the prefix of $\ll$ up to and including $x_i$ is also the chain $[x_{init}]$.

Now imagine we have a prefix $p_i$ of $\ll_x$ which is equal to the versions in the trace of some $x_i$ without intermediate or aborted versions---i.e. which satisfies this lemma. Let $x_j$ be the next element after $x_i$ in $\ll_x$, assuming $x_i$ is not the largest element of $\ll_x$. By axiom \ref{version-order->version-graph}, there exists a path from $x_i \xrightarrow{} x_j$ in $v_x$, and since $x$ is traceable, that path is unique. That path cannot contain any installed versions (other than $x_i$ and $x_j$) without contradicting axiom \ref{version-order->version-graph}. Removing intermediate and aborted versions from the path therefore yields $[x_i, x_j]$, and we know that removing intermediate and aborted versions from the trace of $x_i$ yields $p_i$. Therefore, removing intermediate and aborted versions from the trace of $x_j$ yields the prefix of $\ll_x$ up to and including $x_j$.

By induction, this property holds for all installed $x_j$.
\end{proof}

The version order for $x$, $\ll_x$, provides a total order of installed (e.g. non-intermediate and non-aborted) versions of $x$. Per axiom \ref{version-order->version-graph}, $v_x$ has a path which connects those versions, possibly via intermediate and/or aborted versions. 

\end{journal}

Restricting our histories to traceable objects (e.g., lists) will allow us to directly reason about the version order $\ll$ using the results of individual read operations.


\subsection{A Theory of Mind for Externally-Observed Histories}


When we interact with a database system, the history may not be accessible from outside the database---or perhaps no ``real'' history exists at all. We construct a formal ``theory of mind'' which allows us to reason about potential Adya histories purely from client observations.

We define an \textit{observation} of a system as a set of experimentally-accessible transactions where versions, return values, and committed states may be unknown, and consider the \textit{interpretations} of that observation---the set of all histories which could have resulted in that particular observation.

To be certain that an external observation constitutes an irrefutable proof of an internal isolation anomaly requires that observations have a unique mapping between versions and observed transactions, a notion we call \emph{recoverability}. We provide practical, sufficient conditions to produce recoverable histories.

\subsubsection{Observations}

Imagine a set of single-threaded logical processes which interact with a database system as clients. Each process submits transactions for execution, and may receive information about the return values of the operations in those transactions. What can we tell from the perspective of those client processes?

Recall that an object operation has five components:\\ $f(x, x_i, a) \xrightarrow{} (x_j, r)$ denotes an operation on object $x$ which takes version $x_i$ and combines it with argument $a$ to yield version $x_j$, returning $r$. When a client makes a write, it knows the object $x$ and argument $a$, but (likely) not the versions $x_i$ or $x_j$. If the transaction commits, the client may know the return value $r$, but might not if, for example, a response message is lost by the network.

We define an \textit{observed object operation} as an operation whose versions and return value may be unknown. We write observed operations with a hat: $\hat{w}(x, \_, 3) \xrightarrow{} (\_, nil)$ denotes an observed write of 3 to object $x$, returning $nil$; the versions involved are unknown.
An \textit{observed operation} is either an observed object operation, a commit, or an abort.

An \textit{observed transaction}, written $\hat{T_i}$, is a transaction composed of observed operations. If a client attempts to abort, or does not attempt to commit, a transaction, the observed transaction ends in an abort. If a transaction is known to have committed, it ends in a commit operation. However, when a client \textit{attempts} to commit a transaction, but the result is unknown, e.g. due to a timeout or database crash, we leave the transaction with neither a commit nor abort operation.

An \textit{Observation} $O$ represents the experimentally-accessible information about a database system's behavior. Observations have a set of observed transactions $\hat{T}$. We assume observations include every transaction executed by a database system.
We say that O is \textit{determinate} if every transaction in $\hat{T}$ is either committed or aborted; e.g. there are no indeterminate transactions. Otherwise, $O$ is \textit{indeterminate}.

Consider the set $X_c$ of versions of a traceable object $x$ read by committed transactions in some observation $O$. We denote a single version with a trace longer than any other $x_{longest}$---if there are multiple longest traces, any will do. We say $O$ is \textit{consistent} if for all $x$ in $O$, every $x_i \in X_c$ appears in the trace of $x_{longest}$. Otherwise, $O$ is \textit{inconsistent}. We will find $x_{longest}$ helpful in inferring as much of $\ll$ as possible.

\subsubsection{Interpretations}

Intuitively, an observed operation $\hat{o}$ could be a witness to an ``abstract'' operation $o$ if the two execute the same type of operation on the same key with the same argument, and their return values and versions don't conflict. We capture this correspondence in the notion of \textit{compatibility}.

Consider an operation $o = f(x_i, a) \xrightarrow{} (x_j, r)$ and an observed operation $\hat{o} = \hat{f}(\hat{x_i}, \hat{a}) \xrightarrow{} (\hat{x}_j, \hat{r})$. We say that $o$ is \textit{compatible with} $\hat{o}$ iff:

\begin{itemize}
    \setlength\itemsep{0em}
    \item $\hat{f} = f$
    \item $\hat{a} = a$
    \item $\hat{r}$ is either \textit{unknown} or equal to $r$
    \item $\hat{x_i}$ is either \textit{unknown} or equal to $x_i$
    \item $\hat{x}_j$ is either \textit{unknown} or equal to $x_j$
\end{itemize}

We may now define a notion of compatibility among transactions that builds upon object compatibility. Consider an abstract transaction $T_i$ and an observed transaction $\hat{T_i}$. We say that $T_i$ is \textit{compatible with} $\hat{T_i}$ iff:

\begin{itemize}
    \setlength\itemsep{0em}
    \item They have the same number of object operations.
    \item Every object operation in $T_i$ is compatible with its corresponding object operation in $\hat{T_i}$.
    \item If $T_i$ committed, $\hat{T_i}$ is not aborted, and if $\hat{T_i}$ committed, $T_i$ committed too.
    \item If $T_i$ aborted, $\hat{T_i}$ is not committed, and if $\hat{T_i}$ aborted, $T_i$ aborted too.
\end{itemize}

Finally, we generalize the notion of compatibility to entire histories and observations. Consider a history $H$ and an observation $O$, with transaction sets $T$ and $\hat{T}$ respectively. We say that $H$ is \textit{compatible with} $O$ iff there exists a one-to-one mapping $R$ from $\hat{T}$ to $T$ such that $\forall (\hat{T_i}, T_i) \in R, T_i \textrm{ is compatible with } \hat{T_i}$. We call any $(H, R)$ which satisfies this constraint an \textit{interpretation} of $O$. Given an interpretation, we say that $T_i = R \hat{T}_i$ is the \textit{corresponding transaction} to $\hat{T_i}$, and vice versa.

There may be many histories compatible with a given observation. For instance, an indeterminate observed transaction may either commit or abort in a compatible history. Given two increment transactions $T_1$ and $T_2$ with identical operations, there are two choices of $R$ for any history $H$, corresponding to the two possible orders of $T_1$ and $T_2$. There may also be many observations compatible with a given history: for instance, we could observe transaction $T_1$'s commit, or fail to observe it and label $T_1$ indeterminate.

In each interpretation of an observation, every observed transaction corresponds to a distinct abstract transaction in that interpretation's history, taking into account that we may not know exactly what versions or return values were involved, or whether or not observed transactions actually committed. These definitions of \textit{compatible} formalize an intuitive "theory of mind" for a database: what we think could be going on behind the scenes.

\subsubsection{Recoverability}

Traceability allows us to derive version dependencies, but in order to infer transaction dependencies, we need a way to map between versions and observed transactions. We also need a way to identify aborted and intermediate versions, which means proving which particular \textit{write} in a transaction yielded some version(s). To do this, we exploit the definition of reads, and a property relating versions to observed writes, which we call \textit{recoverability}.

The definition of a read requires that the pre-version, post-version, and return value are all equal. This means for an observed committed read, we know exactly what version it observed---and conversely, given a version, we know which reads definitely observed it.\footnote{Indeterminate reads, of course, may have read different values in different interpretations.} We say an observed transaction $\hat{T}_i$ read $x_i$ when $x_i$ is returned in one of $\hat{T}_i$'s reads. By compatibility,
every corresponding transaction $T_i$ must \textit{also} have read $x_i$.

Writes are more difficult, because in general multiple writes could have resulted in a given version. For example, consider two observed increment operations $\hat{w}_1 = w(x, \_, 1) \xrightarrow{} (\_, nil)$ and $\hat{w}_2 = w(x, \_, 1) \xrightarrow{} (\_, nil)$. Which of these writes resulted in, say, the version 2? It could be \textit{either} $\hat{w}_1$ or $\hat{w}_2$. We cannot construct a single transaction dependency graph for this observation. We could construct a (potentially exponentially large) \textit{set} of dependency graphs, and check each one for cycles, but this seems expensive. To keep our analysis computationally tractable, we restrict ourselves to those cases where we can infer a single write, as follows.

Given an observation $O$ and an object $x$ with some version $x_i$, we say that $x_i$ is \textit{recoverable} iff there is exactly one write $\hat{w}_i$ in $O$ which is compatible with any write leading to $x_i$ in the version graph $v_x$. We call $\hat{w}_i$ \textit{recoverable} as well, and say that $x_i$ must have been written by $\hat{w}_i$. Since there is only one $\hat{w}_i$, there is exactly one transaction $\hat{T}_i$ in $O$ which performed $\hat{w}_i$.

Thanks to compatibility, any interpretation of $O$ has exactly one $w_i$ compatible with $\hat{w}_i$, again performed by a unique transaction $T_i$. When a version is recoverable, we know which single transaction performed it in \textit{every} interpretation.

We say a version $x_i$ is \textit{known-aborted} if it is recoverable to an aborted transaction, \textit{known-committed} if it is recoverable to a committed transaction, and \textit{known-intermediate} if it is recovered to a non-final write. By compatibility, these properties apply not just to an observation $O$, but to every interpretation of $O$.

We say an observation $O$ is \textit{completely recoverable} if every write in $O$ is recoverable. $O$ is \textit{intermediate-recoverable} if every intermediate write in $O$ is recoverable. $O$ is \textit{trace-recoverable} if, for every $x$ in $O$, $x$ is traceable, and every non-initial version in the trace of every committed read of $x$ is recoverable.

We can obtain complete recoverability for a register by choosing unique arguments for writes. Counters and sets are difficult to recover in general: a set like $\{1,2\}$ could have resulted either from a write of 1 or 2.\footnote{We can define a weaker notion of recoverability which identifies all writes in the causal past of some version, but we lack space to discuss it here.} However, restricting observations to a single write per object makes recovery trivial.

\begin{sigmod}

For traceable objects, we can guarantee an observation $O$ is trace-recoverable when $O$ satisfies three criteria:

\begin{enumerate}
    \item Every argument in the observed writes to some object is distinct.
    \item Given a committed read of $x_i$, every argument to every write in the trace of $x_i$ is distinct.
    \item Given a committed read of $x_i$, every write in the trace of $x_i$ has a compatible write in $O$.
\end{enumerate}
\end{sigmod}

\begin{journal}
\journo{}
\begin{lemma}
\label{lem:recoverability}
An observation $O$ is recoverable if all of the following hold for every object $x$ in $O$:

\begin{enumerate}
    \item Every argument in the observed writes to $x$ is distinct.
    \item Every argument in the traces of every known-committed version of $x$ is distinct.
    \item Every write in the trace of every known-committed version of $x$ has a compatible write in $O$.
\end{enumerate}
\end{lemma}

\begin{proof}
Consider a version $x_i \neq x_{init}$ in the trace of some version $x_c$ read by a committed transaction in $O$. To show trace-recoverability, we must prove $x_i$ has exactly one $(\hat{T}_i, \hat{w}_i)$ in $M_x$, the recoverability map.

Since $x$ is traceable, and $x_i$ is not the initial state, there is exactly one write $w_i = w(x, \_, a) \xrightarrow{} (x_i, r)$ which yields $x_i$. By (3), a compatible write $\hat{w}_i$ exists in $O$, and by (1), there is only one $\hat{T}_i$ which performed $\hat{w}_i$. Since $\hat{w}_i$ is compatible with $w_i$ (the only write leading to $x_i$ in $v_x$), $(\hat{T}_i, \hat{w}_i)$ is the corresponding transaction and write to $x_i$.
\end{proof}

\end{journal}

We can ensure the first criterion by picking unique values when we write to the database. We can easily detect violations of the remaining two criteria, and each points to pathological database behavior: if arguments in traces are not distinct, it implies some write was somehow applied \textit{twice}; and if a trace's write has no compatible write in $O$, then it must have manifested from nowhere.

Similar conditions suffice for intermediate-recoverability.

With a model for client observations, interpretations of those observations, and ways to map between versions and observed operations, we are ready to infer the presence of anomalies.

\subsection{Soundness of Elle}

We would like our checker to be \emph{sound}: if it reports an anomaly in an observation, that anomaly should exist in every interpretation of that observation. We would also like it to be \emph{complete}: if an anomaly occurred in an history, we should be able to find it in any observation of that history. In this section, we establish the soundness of \Name{} formally, and show how our approach comes \emph{close} to guaranteeing completeness.

The anomalies identified by Adya et al. can be broadly split into two classes. Some anomalies involve transactions operating on versions that they should not have observed, either because an aborted transaction wrote them or because they were not the final write of some committed transaction. Our soundness proof must show that if one of these anomalies is detected in an observation, it surely occurred in every interpretation of that observation. Others involve a cycle in the dependency graph between transactions; we show that given an observation, we can can construct a dependency graph which is a \emph{subgraph} of every possible history compatible with that observation. If we witness a cycle in the subgraph, it surely exists in any compatible history.

We begin with the first class: non-cycle anomalies.

\subsubsection{Non-Cycle Anomalies}
\label{sec:dirty}

We can use the definition of compatibility, along with properties of traceable objects and recoverability, to infer whether or not an observation implies that every interpretation of that observation contains aborted reads, intermediate reads, or dirty updates.

\begin{sigmod}

\textit{Direct Observation}
\hspace{.1cm}
Consider an observation $O$ with a known-aborted version $x_i$. If $x_i$ is read by an observed committed transaction $\hat{T}_i$, that read must correspond to a committed read of an aborted version in every interpretation of $O$: an aborted read. A similar argument holds for intermediate reads.

\textit{Inconsistent Observations}
\hspace{.2cm}
For traceable objects, we can go further. If an observation $O$ is inconsistent, it contains a committed read of some version $x_i$ which does not appear in the trace of $x_{longest}$. As previously discussed, all committed versions of $x$ must be in the trace of $x_{max}$. At most one of $x_{longest}$ or $x_i$ may be in this trace, so at least one of them must be aborted.

\textit{Via Traces}
\hspace{.2cm}
Consider a committed read of some value $x_c$ whose trace contains a known-aborted version $x_a$. Either $x_c$ is aborted (an aborted read), or a dirty update exists between $x_a$ and $x_c$. An similar argument allows us to identify dirty updates when $x_c$ is the product of a known-committed write. The closer $x_a$ and $x_c$ are in the version graph, the better we can localize the anomaly.

\textit{Completeness}
\hspace{.2cm}
The more recoverable a history is, and the fewer indeterminate transactions it holds, the more non-cycle anomalies we can catch. If an observation is determinate and trace-recoverable, we know exactly which reads committed in every interpretation, and exactly which writes aborted, allowing us to identify every case of aborted read. Finding every dirty update requires complete recoverability.

For an intermediate-recoverable observation $O$, we can identify every intermediate read. We can do the same if $O$ is trace-recoverable. Let $x_i$ be a version read by a committed read in $O$. Trace-recoverability ensures $x_i$ is recoverable to a particular write, and we know from that write's position in its observed transaction whether it was intermediate or not. Compatibility ensures all interpretations agree.

In practice, observations are rarely complete, but as we show in section \ref{sec:implementation}, we typically observe \textit{enough} of a history to detect the presence of non-cycle anomalies.
\end{sigmod}

\begin{journal}
\journo{}

\begin{lemma}
\label{inconsistent->aborted-read}
If an observation $O$ is inconsistent, every interpretation of $O$ exhibits aborted reads.
\end{lemma}

\begin{proof}
If $O$ is inconsistent, then there exists a committed read $r_i$ of some version $x_i$ which is not in the trace of $x_{longest}$. If $x_i$ is aborted, then $r_i$ is an aborted read by definition. Likewise, if $x_{longest}$ is not committed, then our committed read of $x_{longest}$ is an aborted read.

If, on the other hand, both $x_i$ and $x_{longest}$ are committed, then consider lemma \ref{traceable->x_max-trace-all-committed}, which states that every committed version of $x$ must be in the trace of $x_{max}$. This means both $x_i$ and $x_{longest}$ are in the trace of $x_{max}$.

Since $x_i$ is not in the trace of $x_{longest}$, $x_{longest}$ must be in the trace of $x_i$---but this would imply $x_i$ is both observed by a committed read, and has a longer trace than $x_{longest}$, which contradicts the definition of $x_{longest}$. Therefore $x_i$ must be aborted, and $r_i$ is an aborted read.
\end{proof}

\begin{lemma}
\label{traceable+recoverable+aborted-write-immediately-before-read->aborted-read}
If a traceable, recoverable observation $O$ has a committed read of version $x_i$, which, (by recoverability) was the product of a write $\hat{w}_i$, and the transaction $\hat{T}_i$ containing $\hat{w}_i$ is aborted, then every interpretation of $O$ exhibits aborted read.
\end{lemma}

\begin{proof}
Consider an interpretation $(H, R)$ of $O$. By compatibility with $O$, $H$ contains a committed read of $x_i$. By traceability, there is exactly one write $w_i$, in transaction $T_i$, which produced $x_i$. Recoverability implies that there is exactly one transaction ($\hat{T}_i$) in $O$ which is compatible with $T_i$. Since $\hat{T}_i$ aborted, $T_i$ aborted as well. $x_i$ is therefore the product of an aborted write, which satisfies the definition of an aborted read.
\end{proof}

\begin{lemma}
\label{traceable+recoverable+aborted-write-before-committed-write-before-read->dirty-update}
If a traceable, recoverable observation $O$ has a committed read of some version $x_k$ with at least two versions $x_i$ and $x_j$ in the trace of $x_k$, in that order, and recoverability implies that the transaction $\hat{T}_i$ which wrote $x_i$ aborted, and the transaction $\hat{T}_j$ committed, then every interpretation of $O$ exhibits dirty update.
\end{lemma}

\begin{proof}
Consider an interpretation $(H, R)$ of $O$. If $x_j$ is aborted in $H$, then the committed read of $x_j$ in $O$ implies a corresponding committed read of $x_j$ in $H$: an aborted read.

If $x_j$ is committed in $H$, then there must exist some write $w_j$ which produced $x_j$ and committed in $H$. Since $w_i$ aborted in $O$ (and therefore in $H$ as well), $w_i \neq w_j$, and $w_i$ must precede $w_j$ in the trace of $x_j$. Lemma \ref{traceable+recoverable+aborted-write-before-committed-write-before-read->dirty-update} tells us that $H$ must contain a dirty update.
\end{proof}

\begin{lemma}
\label{traceable+recoverable+aborted-write-before-read->dirty-update-or-aborted-read}
If a traceable, recoverable observation $O$ has a committed read of some $x_j$, and some write $w_i$ in the trace of $x_j$ belongs (via recoverability) to an aborted transaction in $O$, then every interpretation of $O$ exhibits either an aborted read or a dirty update.
\end{lemma}

\begin{proof}
Analogous to lemma \ref{traceable+recoverable+aborted-write-immediately-before-read->aborted-read}, every interpretation $(H, R)$ of $O$ contains an aborted version $x_i$ and a committed version $x_j$, such that a path exists between $x_i$ and $x_j$ in $v_x$. By traceability, this path is unique. Therefore there must exist some write between $x_i$ and $x_j$ which takes an uncommitted version to a committed one: a dirty update.
\end{proof}

Traceability and recoverability allow us to map observed transactions to the traces of committed reads. If we observe an aborted write whose consequences are reflected in a committed read, we know that either an aborted read or a dirty update has taken place (or that no interpretations exist at all). If we know the committed or aborted status of the transaction which produced the value we read, or we can find a committed write between the two, we can distinguish between dirty updates and aborted reads.

\textit{Detecting Intermediate Reads}
\hspace{.2cm}

Given a history $H$, Adya et al define an intermediate read (G1b) as a committed transaction $T_j$ which reads a version of object $x$ written by a transaction $T_i$, which was not $T_i$'s final modification of $x$. Given an observation $O$, can we say whether all interpretations $(H, R)$ of $O$ exhibit intermediate reads? In general, we cannot: an observation may not tell us what actual \textit{versions} were produced by a given write. However, for traceable objects, we can infer from a read exactly which updates produced which versions. Recoverability, in turn, lets us identify whether those updates were intermediate or final in their respective transactions.

\begin{lemma}
\label{intermediate-read}
Given a recoverable observation $O$ of a database with a traceable object $x$, if $O$ contains a committed transaction $\hat{T}_j$ which reads $x_i$, and the trace of $x_i$ implies, via recoverability, that $x_i$ was the product of an intermediate write in some observed transaction $\hat{T}_i$, then every interpretation $(H, R)$ of $O$ contains an intermediate read.
\end{lemma}

\begin{proof}
Since $\hat{T}_j$ is committed, its corresponding transaction $T_j$ in every $H$ must be committed too. Likewise, $T_i$'s write of $x_i$ must not be $T_i$'s final write of $x$, in order to be compatible with $\hat{T}_i$. By recoverability, and the fact that $R$ is bijective, $T_i \neq T_j$. Therefore, every $H$ has an intermediate read.
\end{proof}

Collectively, these lemmata establish \Name{}'s soundness in identifying non-cycle anomalies:

\begin{lemma}
\label{non-cycle-sound}
Given a traceable, recoverable observation $O$, if \Name{} infers aborted reads, dirty updates, or intermediate reads, then every interpretation of $O$ exhibits corresponding phenomena.
\end{lemma}

We leave completeness (which requires additional constraints on O) as an exercise to the reader.

\end{journal}

\subsubsection{Dependency Cycles}
\label{sec:cycles}

The remainder of the anomalies identified by Adya et al. are defined in terms of cycles in the dependency graph between transactions. Given an observation $O$, we begin by inferring constraints on the version order $\ll$, then use properties of reads and recoverability to map dependencies on \textit{versions} into dependencies on \textit{transactions}.

\textit{Inferred Version Orders}
\hspace{.2cm}
Consider a intermediate-recoverable observation $O$ of a database composed of traceable objects\footnote{We can also derive weaker constraints on the version order from non-traceable objects, which we leave as an exercise for the reader.}, and an interpretation $(H, R)$ of $O$. We wish to show that we can derive part of $H$'s version order $\ll$ from $O$ alone, with minimal constraints on $H$ and $R$. Traceability allows us to recover a prefix of $\ll_x$ from any installed $x_i$ in $H$, assuming we know which transactions in the trace of $x_i$ committed, and which aborted. Let us \textit{assume} $H$ does not contain aborted reads, intermediate reads, or dirty updates. We call such a history \textit{clean}.

Given $O$, which version of $x$ should we use to recover $\ll_x$? Ideally, we would have $x_{max}$. However, there could be multiple interpretations of $O$ with \textit{distinct} $x_{max}$. Instead, we take a version $x_f$ read by a transaction $\hat{T}_f$ such that:

\begin{itemize}
    \setlength\itemsep{0em}
    \item $\hat{T}_f$ is committed.
    \item $\hat{T}_f$ read $x_f$ before performing any writes to $x$.
    \item No other version of $x$ satisfying the above properties has a longer trace than $x_f$.
\end{itemize}

\begin{sigmod}

We use $x_f$ to obtain an \emph{inferred version order} $<_x$ that is consistent with $\ll_x$, as follows. First, we know that $x_f$ corresponds to an installed version of $x$ in $H$ because $H$ contains no intermediate or aborted reads. By a similar argument, we also know that every version of $x$ in the trace of $x_f$ was written by a committed transaction. Therefore, if we remove the intermediate versions in the trace of $x_f$ (which we know, thanks to intermediate-recoverability), we are left with a total order over committed versions that corresponds directly to the prefix of $\ll_x$ up to and including $x_f$.

We define $<$ as the union of $<_x$ for all objects $x$.

\end{sigmod}

\begin{journal}
\journo{}

\begin{lemma}
\label{x_f-is-installed}
If $O$ is an observation, and $(H, R)$ is an interpretation of $O$ which has neither aborted reads nor intermediate reads, $x_f$ is an installed version in $H$.
\end{lemma}

\begin{proof}
Since $O$ contains a transaction $\hat{T}_f$ reading $x_f$, there must be a compatible $T_f$ in $H$ which is also committed and read $x_f$ before performing any writes to $x$. Since $T_f$ is committed, and $H$ has neither intermediate nor aborted reads, we know that $x_f$ must have been the result of a final write by some committed transaction in $H$.
\end{proof}

\begin{lemma}
\label{x_f-trace-is-all-committed}
If $O$ is an observation, and $(H, R)$ is a clean interpretation of $O$, then every version in the trace of $x_f$ was written by a committed transaction in $H$.
\end{lemma}

\begin{proof}
Some of the transactions in the trace of $x_f$ might be indeterminate in $O$. However, we know from lemma \ref{x_f-is-installed} that $x_f$ is an installed (and therefore committed) version. If an uncommitted write \textit{were} in the trace of $x_f$, then there must exist a committed transaction which performed a write on top of uncommitted state. This would constitute a dirty update: a contradiction. Therefore, every version in the trace of $x_f$ was written by a committed transaction in $H$.
\end{proof}

\begin{lemma}
\label{obs->ll-prefix}
If $O$ is an observation, and $(H, R)$ is a clean interpretation of $O$, then for every object $x$ in $O$, removing intermediate versions from the trace of $x_f$ yields $<_x$: the prefix chain of $\ll_x$ up to and including $x_f$.
\end{lemma}

\begin{proof}
Lemma \ref{x_f-is-installed} tells us that $x_f$ is an installed version in $H$, and lemma \ref{x_f-trace-is-all-committed} tells us that every version in the trace of $x_f$ was written by a committed transaction in $H$. Lemma \ref{x_i->ll_x-prefix} states that removing intermediate and aborted versions from the trace of $x_f$ yields the prefix of $\ll_x$ up to and including $x_f$. We know this trace has no aborted versions, so all that must be done is to remove intermediate versions.
\end{proof}

Note that this lemma holds regardless of our specific choice of $(H, R)$. Given any observation, we can reconstruct a prefix of $\ll_x$ which holds for \textit{all} clean interpretations.

Next, we define an \textit{inferred version order} $<$ by combining the prefix chains from lemma \ref{obs->ll-prefix}: $< \: \equiv \: <_x \cup <_y \cup \ldots$, where $<_x \cup <_y$ means the chain whose members are the union of members in $<_x$ and $<_y$, and whose order is the union of orders in $<_x$ and $<_y$.

\end{journal}

\textit{Inferred Serialization Graphs}
\hspace{.2cm}
Given an intermediate-recoverable observation $O$ of a database of traceable objects, we can infer a chain of versions $<_x$ which is a prefix of $\ll_x$, for every object $x$ in $O$. If $O$ is trace-recoverable, we can map every version in $<$ to a particular write in $O$ which produced it, such that the corresponding write in every interpretation of $O$ produced that same version. Using these relationships, we define \textit{inferred dependencies} between pairs of transactions $\hat{T}_i$ and $\hat{T}_j$ in $O$ as follows:

\begin{itemize}
    \setlength\itemsep{0em}
    \item \textbf{Direct inferred write-write dependency}. A transaction $\hat{T}_j$ \textit{directly inferred-ww-depends} on $\hat{T}_i$ if $\hat{T}_i$ performs a final write of a version $x_i$ of $x$, and $\hat{T}_j$ performs a final write resulting in $x$'s next version $x_j$, by $<$.
    
    \item \textbf{Direct inferred write-read dependency}. A transaction $\hat{T}_j$ \textit{directly inferred-wr-depends} on $\hat{T}_i$ if $\hat{T}_i$ performs a final write of a version $x_i$ in $<$, and $\hat{T}_j$ reads $x_i$.
    
    \item \textbf{Direct inferred read-write dependency}. A transaction $\hat{T}_j$ \textit{directly inferred-rw-depends} on $\hat{T}_i$ if $\hat{T}_i$ reads version $x_i$ of $x$, and $\hat{T}_j$ performs a final write of $x$'s next version in $<$.
\end{itemize}

Unlike Adya et al's definitions, we don't require that a transaction \textit{install} some $x_i$, because an indeterminate transaction in $O$ might be committed in \textit{interpretations} of $O$, and have corresponding dependency edges there. Instead, we rely on the fact that $<$ only relates installed versions (in clean interpretations).

An \textit{Inferred Direct Serialization Graph}, or \textit{IDSG}, is a graph of dependencies between observed transactions. The IDSG for an observation $O$ is denoted IDSG($O$). If $\hat{T}_j$ inferred-ww-depends on $\hat{T}_i$, there exists an edge labeled \textit{ww} from $\hat{T}_i$ to $\hat{T}_j$ in IDSG($O$), and similarly for inferred-wr and inferred-rw-dependencies.

All that remains is to show that for every clean interpretation $(H, R)$ of an observation, IDSG($O$) is (in some sense) a subgraph of DSG($H$). However, the IDSG and DSG are graphs over different types of transactions; we need the bijection $R$ to translate between them. Given a relation $R$ and a graph $G$, we write $R \diamond G$ to denote $G$ with each vertex $v$ replaced by $Rv$.

\begin{sigmod}
The soundness proof for \Name{} first establishes that
for every clean interpretation $(H, R)$ of a trace-recoverable observation $O$, $R \diamond IDSG(O)$ is a subgraph of $DSG(H)$.
The proof proceeds by cases showing that for each class of dependency, if a given edge exists in the IDSG, it surely exists in every compatible DSG.  We omit these details, which are straightforward, in this paper.

\end{sigmod}

\begin{journal}
\journo{}
\begin{lemma}
\label{IDSG-subgraph-of-DSG}
For every clean interpretation $(H, R)$ of an observation $O$, $R \diamond IDSG(O)$ is a subgraph of $DSG(H)$.
\end{lemma}

\begin{proof}
It suffices to show that every edge in $R \diamond IDSG(O)$ is also in $DSG(H)$. We consider a pair of observed transactions $\hat{T}_i$ and $\hat{T}_j$ in $O$, and their corresponding transactions in $H$: $T_i = R \hat{T}_i$ and $T_j = R \hat{T}_j$.

\begin{itemize}
\item[ww] If $\hat{T}_j$ directly inferred-ww-depends on $\hat{T}_i$, then $T_j$ directly ww-depends on $T_i$.

\begin{proof}
Per the definition of inferred-ww-depends, $\hat{T}_i$ performed a final write of some version $x_a$, and $\hat{T_j}$ performed a final write of some $x_b$, such that $x_a$ immediately precedes $x_b$ in $<$. By lemma \ref{obs->ll-prefix}, $<_x$ is a prefix of $\ll_x$, which means $x_a$ also directly precedes $x_b$ in $\ll$. Since $\ll$ relates installed versions, $x_a$ and $x_b$ are both installed. Recoverability ensures that there is only one write in $O$ which could produce $x_a$ (and $x_b$, respectively). Since $R$ relates compatible transactions, $T_i = R \hat{T}_i$, and likewise for $T_j$. Since $x_a$ directly precedes $x_b$ in $\ll$, and $T_i$ installed $x_a$, and $T_j$ installed $x_b$, $T_j$ directly ww-depends on $T_i$.
\end{proof}

\item[wr] If $\hat{T}_j$ directly inferred-wr-depends on $\hat{T}_i$, then $T_j$ directly wr-depends on $T_i$.

\begin{proof}
By the definition of inferred-wr-depends, $\hat{T}_i$ performs a final write of a version $x_i$ in $<$, and $\hat{T}_j$ reads $x_i$. Since $x_i$ is in $<$, it is committed in $H$. Compatibility and recoverability imply that $T_i$'s write producing $x_i$ was final; $T_i$ therefore installs $x_i$ in $H$. Compatibility implies $T_j$ read $x_i$. Therefore, $T_i$ directly wr-depends on $T_j$.
\end{proof}

\item[rw] If $\hat{T}_j$ directly inferred-rw-depends on $\hat{T}_i$, then $T_j$ directly rw-depends on $T_i$.

\begin{proof}
By the definition of inferred-rw-depends, $\hat{T}_i$ contains a read of $x_i$ and $\hat{T}_j$ performs a final write of $x_j$, such that $x_i$ immediately precedes $x_j$ in $<$. Via $R$, $T_i$ also contains a read of $x_i$. By recoverability, $T_j$ also performs a final write of $x_j$. Since $<$ contains $x_j$, $T_j$ must have committed, which means $T_j$ installed $x_j$. By lemma \ref{obs->ll-prefix}, $<_x$ is a prefix of $\ll_x$, which means $x_i$ directly precedes $x_j$ in $\ll$ as well. Therefore, $T_j$ directly rw-depends on $T_i$.
\end{proof}
\end{itemize}

Since every edge $\hat{T}_i \xrightarrow{} \hat{T}_j$ in $IDSG(O)$ has a corresponding edge $T_i \xrightarrow{} T_j$ in $DSG(H)$, $R \diamond IDSG(O)$ is a subgraph of $DSG(H)$.
\end{proof}


\subsubsection{Inferred Anomalies}

Lemma~\ref{IDSG-subgraph-of-DSG} tells us that every cycle in IDSG($O$) has a corresponding cycle in every clean interpretation of $O$. 
\end{journal}

For every anomaly defined in terms of cycles on a DSG (e.g. G0, G1c, G-Single, G2, \ldots), we can now define a corresponding anomaly on an IDSG. If we detect that anomaly in IDSG($O$), its corresponding anomaly must be present in every clean interpretation of $O$ as well!

We present a soundness theorem for \Name{} below:

\newtheorem{theorem}{Theorem}

\begin{theorem}
\label{sound}
Given a trace-recoverable observation $O$, if \Name{} infers aborted reads, dirty updates, or intermediate reads, then every interpretation of $O$ exhibits corresponding phenomena. If \Name{} infers a cycle anomaly, then every clean interpretation of $O$ exhibits corresponding phenomena.
\end{theorem}

\begin{journal}
\journo{}

\begin{proof}
The first half of this theorem follows directly from Lemma~\ref{non-cycle-sound}. The second follows from Lemma~\ref{IDSG-subgraph-of-DSG}.
\end{proof}
\end{journal}

\textit{Unclean Interpretations}
\hspace{.2cm}
What of unclean interpretations, like those with aborted reads or dirty updates? If those occurred, the trace of a version read by a committed transaction could cause us to infer a version order $<_x$ which includes uncommitted versions, and is not a prefix of $\ll_x$. A clean interpretation could have cycles absent from an unclean interpretation, and vice versa.

Phenomena like aborted reads and dirty updates are, in an informal sense, ``worse'' than dependency cycles like G1c and G2. If every interpretation of an observation must exhibit aborted reads, the question of whether it also exhibits anti-dependency cycles is not as pressing! And if some interpretations exist which \textit{don't} contain aborted reads, but all of those exhibit anti-dependency cycles, we can choose to give the system the benefit of the doubt, and say that it definitely exhibits G2, but may not exhibit aborted reads.

\textbf{Completeness}
\hspace{.2cm}
The more determinate transactions an observation contains, the more likely we are to definitively detect anomalies. In special cases (e.g. when $O$ is determinate, completely-recoverable, etc.), we can prove completeness. In practice, we typically fail to observe the results of some transactions, and must fall back on probabilistic arguments. In section \ref{sec:implementation} we offer experimental evidence that \Name{} is complete enough to detect anomalies in real databases.

\color{black}

\section{Inferring Additional Dependencies}
\label{sec:additional}

We have argued that \Name{} can infer transaction dependencies based on traceability and recoverability. In this section, we suggest additional techniques for inferring the relationships between transactions and versions.

\subsection{Transaction Dependencies}

In addition to dependencies on values, we can infer additional dependencies purely from the concurrency structure of a history. For instance, if process A performs $T_1$ then $T_2$, we can infer that $T_1 <_p T_2$. These dependencies encode a constraint akin to sequential consistency: each process should (independently) observe a logically monotonic view of the database. We can use these dependencies to strengthen any consistency model testable via cycle detection. For instance, Berenson et al’s definition of snapshot isolation~\cite{bernstein-cc} does not require that transaction start timestamps proceed in any particular order, which means that a single process could observe, then un-observe, a write. If we augment the dependency graph with per-process orders, we can identify these anomalies, distinguishing between SI and strong session SI~\cite{strong-session-si}.

Similarly, serializability makes no reference to real-time constraints: it is legal, under Adya's formalism, for every read-only transaction to return an initial, empty state of the database, or to discard every write-only transaction by ordering it after all reads. Strict serializability~\cite{linearizability} enforces a real-time order: if transaction $T_1$ completes before $T_2$ begins, $T_2$ must appear to take effect after $T_1$. We can compute a transitive reduction of the realtime precedence order in $O(n \cdot p)$ time, where $n$ is the number of operations in the history, and $p$ is the number of concurrent processes, and use it to detect additional cycles.

Some snapshot-isolated databases expose transaction start and commit timestamps to clients. Where this information is available, we can use it to infer the time-precedes order used in Adya's formalization of snapshot isolation~\cite{adya-thesis}, and construct a start-ordered serialization graph.

\subsection{Version Dependencies}

Traceability on $x$ allows us to infer a prefix of the version order $<_x$---but this does not mean that non-traceable objects are useless! If we relax $<_x$ to be a partial order, rather than total, and make some small, independent assumptions about the behavior of individual objects, we can recover enough version ordering information to detect cyclic anomalies on less-informative datatypes, such as registers or sets.

For instance, if we assume that the initial version $x_{init}$ is never reachable via any write, we can infer $x_{init} <_x x_i$ for every $x_i$ other than $x_{init}$. With registers, for example, we know that 1, 2, 3, etc. must all follow $nil$. When the number of versions per object is small (or when databases have a habit of incorrectly returning $nil$), this trivial inference can be sufficient to find real-world anomalies.

If we assume that writes follow reads within a single transaction, we can link versions together whenever a transaction reads, then writes, the same key, and that write is recoverable. For instance, $T_1 = r(x_i), w(x_j)$ allows us to infer $x_i <_x x_j$.

Many databases claim that each record is independently linearizable, or sequentially consistent. After computing the process or real-time precedence orders, we can use those transaction relationships to infer \textit{version} dependencies. If a transaction finishes writing or reading a linearizable object $x$ at $x_i$, then another transaction precedes to write or read $x_j$, we can infer (on the basis of per-key linearizability) that $x_i <_x x_j$.

Where databases expose version metadata to clients, we can use that metadata to construct version dependency graphs directly, rather than inferring the version order from values.

Since we can use transaction dependencies to infer version dependencies, and version dependencies to infer transaction dependencies, we can iterate these procedures to infer increasingly complete dependency graphs, up to some fixed point. We can then use the resulting transaction graph to identify anomalies.

\section{Finding Counterexamples}
These techniques allow us to identify several types of dependencies between transactions: write-read, write-write, and read-write relationships on successive versions of a single object, process and real-time orders derived from the concurrency structure of the history, and version and snapshot metadata orders where databases offer them. We take the union of these dependency graphs, with each edge labeled with its dependency relationship(s), and search for cycles with particular properties.

\begin{itemize}
    \item \textbf{G0}: A cycle comprised entirely of write-write edges.
    \item \textbf{G1c}: A cycle comprised of write-write or write-read edges.
    \item \textbf{G-single}: A cycle with exactly one read-write edge.
    \item \textbf{G2}: A cycle with one or more read-write edges.
\end{itemize}

Optionally, we may expand these definitions to allow process, realtime, version, and/or timestamp dependencies to count towards a cycle.

To find a cycle, we apply Tarjan’s algorithm to identify strongly connected components~\cite{tarjan}. Within each graph component, we apply breadth-first search to identify a short cycle.

To find G0 anomalies, we restrict the graph to only write-write edges, which ensures that any cycle we find is purely comprised of write dependencies. For G1c, we select only write-write and write-read edges. G-single is trickier, because it requires exactly one read-write edge. We partition the dependency graph into two subgraphs: one with, and one without read-write edges. We find strongly connected components in the full graph, but for finding a cycle, we begin with a node in the read-write subgraph, follow exactly one read-write edge, then attempt to complete the cycle using only write-write and write-read edges. This allows us to identify cycles with exactly one read-write edge, should one exist.

These cycles can be presented to the user as a witness of an anomaly. We examine the graph edges between each pair of transactions, and use those relationships to construct a human-readable explanation for the cycle, and why it implies a contradiction.

\subsection{Additional Anomalies}

As described in section \ref{sec:dirty}, we can exploit recoverability and traceability to directly detect aborted read, intermediate read, and dirty update. In addition, there are phenomena which Adya et al.’s formalism does not admit, but which we believe (having observed them in real databases) warrant special verification:

\begin{itemize}
    \item \textbf{Garbage reads}: A read observes a value which was never written.
    \item \textbf{Duplicate writes}: The trace of a committed read version contains a write of the same argument multiple times.
    \item \textbf{Internal inconsistency}: A transaction reads some value of an object which is incompatible with its own prior reads and writes.
\end{itemize}

Garbage reads may arise due to client, network, or database corruption, errors in serialization or deserialization, etc. Duplicate writes can occur when a client or database retries an append operation; with registers, duplicate writes can manifest as G1c or G2 anomalies. Internal inconsistencies can be caused by improper isolation, or by optimistic concurrency control which fails to apply a transaction's writes to its local snapshot.

\section{Implementation and Case Study}
\label{sec:implementation}

\begin{figure*}[htb]

\centering
\begin{BVerbatim}
    
Let:
  T1 = {:value [[:append 250 10] [:r 253 [1 3 4]] [:r 255 [2 3 4 5]] [:append 256 3]], ...}
  T2 = {:value [[:append 255 8] [:r 253 [1 3 4]]], ...}
  T3 = {:value [[:append 256 4] [:r 255 [2 3 4 5 8]] [:r 256 [1 2 4]] [:r 253 [1 3 4]]], ...}

Then:
  - T1 < T2, because T1 did not observe T2's append of 8 to 255.
  - T2 < T3, because T3 observed T2's append of 8 to key 255.
  - However, T3 < T1, because T1 appended 3 after T3 appended 4 to 256: a contradiction!
  
\end{BVerbatim}
\caption{A textual explanation of an experimentally obtained real-time G-single cycle, as presented by our checker.}
\label{fig:witness}
\end{figure*}

\begin{figure}[htb!]
\includegraphics[width=\columnwidth]{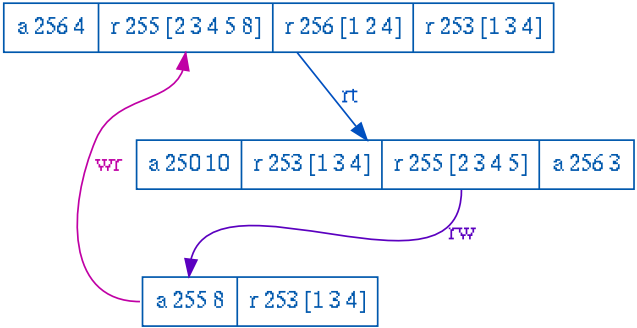}
\caption{The same cycle, as plotted by our checker. Arrows show dependencies between transactions: wr denotes a read dependency, rw denotes an anti-dependency, and rt denotes a real-time ordering.}
\label{fig:witness-plot}
\end{figure}

We have implemented \Name{} as a checker in the open-source distributed systems testing framework \Jepsen{} \cite{jepsen-library} and applied it to four distributed systems, including SQL, document, and graph databases. \Name{} revealed anomalies in every system we tested, including G2, G-single, G1a, lost updates, cyclic version dependencies, and internal inconsistency. Almost all of these anomalies were previously unknown. We have also demonstrated, as a part of \Name{}'s test suite, that \Name{} can identify G0, G1b, and G1c anomalies, as well as anomalies involving real-time and process orders.

\Name{} is straightforward to run against real-world databases. Most transactional databases offer some kind of list with append. The SQL standard's \texttt{CONCAT} function and the \texttt{TEXT} datatype are a natural choice for encoding lists, e.g. as comma-separated strings. Some SQL databases, like \Postgres{}, offer JSON collection types. Document stores typically offer native support for ordered collections. Even systems which only offer registers can \textit{emulate} lists by performing a read followed by a write.

While list-append gives us the most precise inference of anomalies, we can use the inference rules discussed in section \ref{sec:additional} to analyze systems without support for lists. Wide rows in \Cassandra{} and predicates in SQL are a natural fit for sets. Many systems have a notion of an object version number or counter datatype: we can detect cycles in both using \Name{}. Even systems which offer only read-write registers allow us to infer write-read dependencies directly, and version orders can be (partially) inferred by write-follows-read, process, and real-time orders.

In all our tests, we generated transactions of varying length (typically 1-10 operations) comprised of random reads and writes over a handful of objects. We performed anywhere from one to 1024 writes per object; fewer writes per object stresses codepaths involved in the creation of fresh database objects, and more writes per object allows the detection of anomalies over longer time periods.

We ran ~10-30 client threads across 5 to 9 nodes, depending on the particular database under test. When a client thread times out while committing a transaction (as is typical for fault-injection tests) \Jepsen{} spawns a new logical process for that thread to execute. This causes the logical concurrency of tests to rise over time. Tens of thousands of logically concurrent transactions are not uncommon.

Our implementation takes an expected consistency model (e.g. strict-serializable) and automatically detects and reports anomalies as data structures, visualizations, and human-verifiable explanations of each cycle. For example, consider the G-single anomaly in Figures~\ref{fig:witness} and \ref{fig:witness-plot}.

\subsection{TiDB}

\TiDB{}~\cite{tidb} is an SQL database which claims to provide snapshot isolation, based on Google's \Percolator{}~\cite{percolator}. We tested list append with SQL \texttt{CONCAT} over \texttt{TEXT} fields, and found that versions 2.1.7 through 3.0.0-beta.1 exhibited frequent anomalies---even in the absence of faults. For example, we observed the following trio of transactions:

\begin{itemize}
    \setlength\itemsep{0em}
    \item[$T_1$:] \texttt{r(34, [2, 1]), append(36, 5), append(34, 4)}
    \item[$T_2$:] \texttt{append(34, 5)}
    \item[$T_3$:] \texttt{r(34, [2, 1, 5, 4])}
\end{itemize}

$T_1$ did not observe $T_2$'s append of 5 to key 34, so $T_2$ must rw-depend on $T_1$. However, $T_3$'s read implies $T_1$'s append of 4 to key 34 followed $T_2$'s append of 5, so $T_1$ ww-depends on $T_2$. This cycle contains exactly one anti-dependency edge, so it is a case of G-single: read skew. We also found numerous cases of \textit{inconsistent} observations (implying aborted reads) as well as lost updates.

These cases stemmed from an automated transaction retry mechanism: when one transaction conflicted with another, \TiDB{} simply re-applied the transaction's writes again, ignoring the conflict. This feature was enabled by default. Turning it off revealed the existence of a second, undocumented retry mechanism, also enabled by default. Version 3.0.0-rc2 resolved these issues by disabling both retry mechanisms by default.

Furthermore, \TiDB{}'s engineers claimed that \texttt{select ... for update} prevented write skew. \Name{} demonstrated that G2 anomalies including write skew were still possible, even when all reads used \texttt{select ... for update}. \TiDB{}'s locking mechanism could not express a lock on an object which hadn't been created yet, which meant that freshly inserted rows were not subject to concurrency control. \TiDB{} has documented this limitation.

\subsection{YugaByte DB}

\YugaByte{}~\cite{yugabyte} is a serializable SQL database based on Google's \Spanner{}~\cite{spanner}. We evaluated version 1.3.1 using \texttt{CONCAT} over \texttt{TEXT} fields, identified either by primary or secondary keys, both with and without indices. We found that when master nodes crashed, paused, or otherwise became unavailable to tablet servers, those tablet servers could exhibit a handful of G2-item anomalies. For instance, this cycle (condensed for clarity), shows two transactions which fail to observe each other's appends:

\begin{itemize}
    \setlength\itemsep{0em}
    \item[$T_1$:] \ldots~\texttt{append(3, 837)}~\ldots~\texttt{r(4, [}~\ldots~\texttt{874, 877, 883])}
    \item[$T_1$:] \ldots~\texttt{append(4, 885)}~\ldots~\texttt{r(3, [}~\ldots~\texttt{831, 833, 836])}
\end{itemize}

Every cycle we found involved multiple anti-dependencies; we observed no cases of G-single, G1, or G0. \YugaByte{}'s engineers traced this behavior to a race condition: after a leader election, a fresh master server briefly advertised an empty \textit{capabilities set} to tablet servers. When a tablet server observed that empty capabilities set, it caused every subsequent RPC call to include a read timestamp. \YugaByte{} should have ignored those read timestamps for serializable transactions, but did not, allowing transactions to read from inappropriate logical times. This issue was fixed in 1.3.1.2-b1.

\begin{figure*}[htb!]
\includegraphics[width=\textwidth]{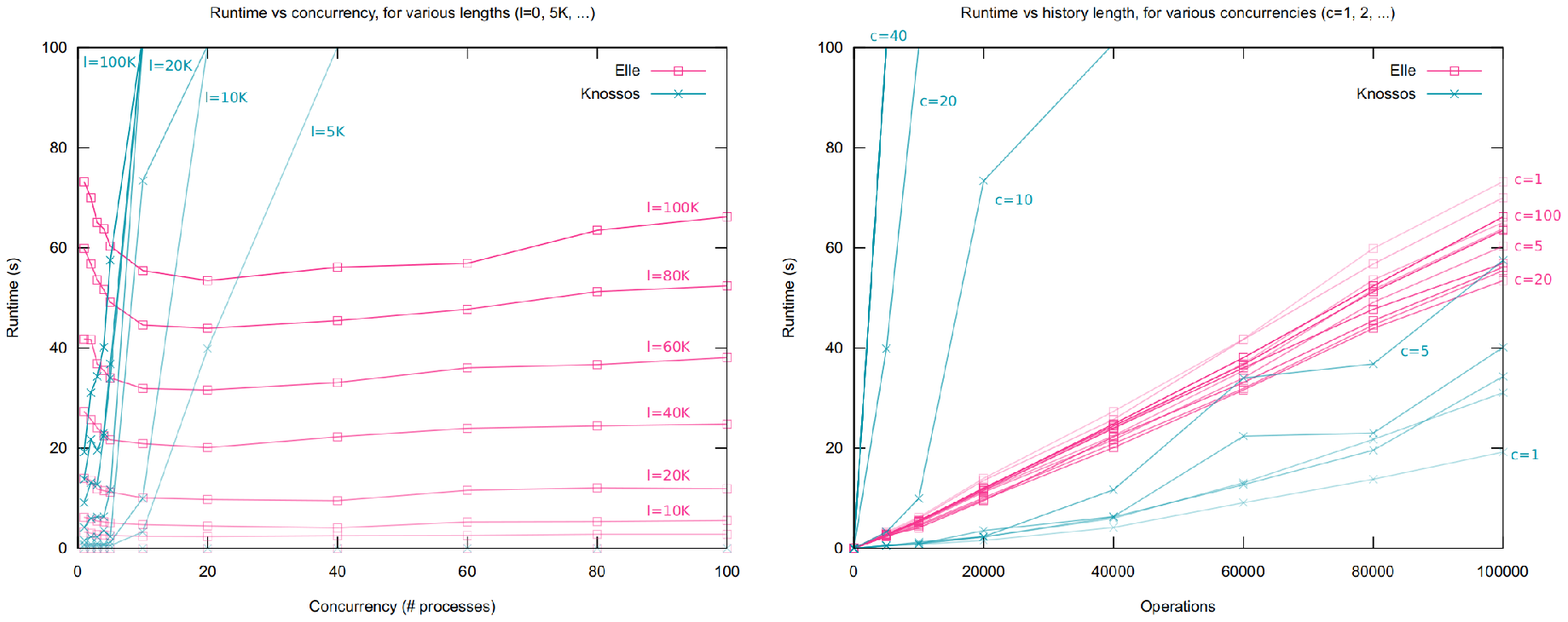}
\caption{Performance of \Name{} vs \Knossos{}.}
\label{fig:performance}
\end{figure*}

\subsection{FaunaDB}

\FaunaDB{}~\cite{fauna} is a deterministic, strict-serializable document database based on \Calvin{}~\cite{calvin}. It offers native list datatypes, but the client we used had no list-append function---we used strings with \texttt{concat} instead. While \FaunaDB{} claimed to provide (up to) strict serializability, we detected \textit{internal} inconsistencies in version 2.6.0, where a single transaction failed to observe its own prior writes:

\begin{itemize}
    \setlength\itemsep{0em}
    \item[$T_1$:] \texttt{append(0, 6), r(0, nil)}
\end{itemize}

These internal inconsistencies also caused \Name{} to infer G2 anomalies. Internal anomalies occurred frequently, under low contention, in clusters without any faults. However, they were limited to index reads. Fauna believes this could be a bug in which coordinators fail to apply tentative writes to a transaction's view of an index.

\subsection{Dgraph}

\Dgraph{}~\cite{dgraph} is a graph database with a homegrown transaction protocol influenced by Shacham, Ohad et al.~\cite{omid-reloaded} \Dgraph{}'s data model is a set of entity-attribute-value triples, and it has no native list datatype. However, it \textit{does} lend itself naturally to registers, which we analyzed with \Name{}. We evaluated \Dgraph{} version 1.1.1, which claimed to offer snapshot isolation, plus per-key linearizability.

Like \FaunaDB{}, \Dgraph{} transactions failed to provide internal consistency under normal operation: reads would fail to observe previously read (or written!) values. This transaction, for instance, set key 10 to 2, then read an earlier value of 1.

\begin{itemize}
    \setlength\itemsep{0em}
    \item[$T_1$:] \texttt{w(10, 2), r(10, 1)}
\end{itemize}

To find cycles over registers, we allowed \Name{} to infer partial version orders from the initial state, from writes-follow-reads within individual transactions, and (since \Dgraph{} claims linearizable keys) from the real-time order of operations. These inferred dependencies were often \textit{cyclic}---here, transaction $T_1$ finished writing key 540 a full three seconds before $T_2$ began, but $T_2$ failed to observe that write:

\begin{itemize}
    \setlength\itemsep{0em}
    \item[$T_1$:] \texttt{r(541, nil), w(540, 2)}
    \item[$T_2$:] \texttt{r(540, nil), w(544, 1)}
\end{itemize}

\Name{} automatically reports and discards these inconsistent version orders, to avoid generating trivial cycles, but it went on to identify numerous instances of read skew, both with and without real-time edges:

\begin{itemize}
    \setlength\itemsep{0em}
    \item[$T_1$:] \texttt{r(2432, 10), r(2434, nil)}
    \item[$T_2$:] \texttt{w(2434, 10)}
    \item[$T_3$:] \texttt{w(2432, 10), r(2434, 10)}
\end{itemize}

These cycles stemmed from a family of bugs in \Dgraph{} related to shard migration: transactions could read from freshly-migrated shards without any data in them, returning $nil$. Dgraph Labs is investigating these issues.

\subsection{Performance}

\Name{}'s performance on real-world workloads was excellent; where \Knossos{} (\Jepsen{}'s main linearizability checker) often timed out or ran out of memory after a few hundred transactions, \Name{} was able to check histories of hundreds of thousands of transactions in tens of seconds. To confirm this behavior experimentally, we designed a history generator which simulates clients interacting with an in-memory serializable-snapshot-isolated database, and analyzed those histories with both \Name{} and \Knossos{}.

Our histories were composed of randomly generated transactions performing one to five operations each, interacting with any of 100 possible objects at any point in time. We performed 100 appends per object. We generated histories of different lengths, and with varying numbers of concurrent processes, and measured both \Name{} and \Knossos{}' runtime. Since many \Knossos{} runs involved search spaces on the order of $10^{24}$, we capped runtimes at 100 seconds. All tests were performed on a 24-core Xeon with 128 GB of ram.

As figure \ref{fig:performance} shows, \Knossos{}' runtime rises dramatically with concurrency: given $c$ concurrent transactions, the number of permutations to evaluate is $c!$. Symmetries and pruning reduce the state space somewhat, but the problem remains fundamentally NP-complete. With 40+ concurrent processes, even histories of 5000 transactions were (generally) uncheckable in reasonable time frames. Of course, runtime rises with history length as well.

\Name{} does not exhibit \Knossos{}' exponential runtimes: it is primarily linear in the length of a history. Building indices, checking for consistent orders, looking for internal and aborted reads, constructing the inferred serialization graph, and detecting cycles are all linear-time operations. Unlike \Knossos{}, concurrency does not have a strong impact on \Name{}. With only one process, every transaction commits. As concurrency rises, some transactions abort due to conflicts, which mildly reduces the number of transactions we have to analyze. At high concurrency, more transactions interact with the same versions, and we infer more dependencies.

\section{Related Work}
As we discuss in Section \ref{sec:intro}, there has been a significant amount of work on history checkers in the concurrent programming community. As early as 1993, Wing \& Gong~\cite{wing-concurrent} simulated executions of linearizable objects to record concurrent histories, and described a checker algorithm which could search for bugs in those histories. \LineUp{}~\cite{lineup}, \Knossos{}~\cite{knossos}, and Lowe's linearizability checker~\cite{lowe} follow similar strategies. Gibbons \& Korach showed~\cite{gibbons} that sequential consistency checking was NP-complete via reduction to SAT.

Generating random operations, applying them to some implementation of a datatype, and checking that the resulting history obeys certain invariants is a key concept in \textit{generative}, or \textit{property-based} testing. Perhaps the most well-known implementation of this technique is \QuickCheck{}~\cite{quickcheck}, and \Jepsen{} applies a similar approach to distributed systems~\cite{jepsen-reports}. Majumdar \& Niksic argued probabilistically for the effectiveness of this randomized testing approach~\cite{random-testing}, which helps explain why our technique finds bugs.

Brutschy et al. propose both a static~\cite{brutschy-static} and a trace-based dynamic~\cite{brutschy-dynamic} analysis to find serializability violations in programs run atop weakly-consistent stores. Quite recently, Biswas \& Enea provided polynomial-time checkers for read committed, read atomic, and causal consistency, as well as exponential-time checkers for prefix consistency, snapshot isolation, and serializability.~\cite{biswas}

Using graphs to model dependencies among transactions has a long history in the database literature.  
The dependency graph model was first proposed by Bernstein~\cite{formalaspects,bernstein-cc} and later refined by Adya~\cite{adya-gil,adya-thesis}. Dependency graphs have been applied to the problem of safely running transactions at a mix of isolation levels~\cite{allocating} and to the problem of runtime concurrency control\cite{serializable-snapshot,dgcc}, in addition to reasoning formally about isolation levels and anomalous histories.

As attractive as dependency graphs may be as a foundation for database testing, they model orderings
among object versions and operations that are not necessarily visible to external clients. Instead of defining isolation levels in terms of internal operations, some declarative definitions
of isolation levels~\cite{cerone,ec} are based upon a pair of compatible dependency relations: a \emph{visibility relation} capturing the order in which writes are visible to transactions and an \emph{arbitration relation} capturing the order in which writes are committed.

The client-centric formalism of Crooks et al.~\cite{seeing} 
goes a step further, redefining consistency levels strictly in terms of client-observable states.  While both approaches, like ours, enable reasoning about existing isolation levels from the \emph{outside} of the database implementation, our goal is somewhat different.  We wish instead to provide a faithful \emph{mapping} between externally-observable events and Adya's data-centric definitions, which have become a \emph{lingua franca} in the database community. In so doing, we hope to build a bridge between two decades of scholarship on dependency graphs and emerging techniques for black-box database testing. 

\section{Future Work \& Conclusions}
\textit{Future Work}
\hspace{.3cm}
There are some well-known anomalies, like long fork, which \Name{} detects but tags as G2. We believe it should be possible to provide more specific hints to users about what anomalies are present. Ideally, we would like to tell a user exactly which isolation levels a given history does and does not satisfy.

Our approach ignores predicates and deals only in individual objects; we cannot distinguish between repeatable read and serializability. Nor can we detect anomalies like predicate-many-preceders. We would like to extend our model to represent predicates, and prove how to infer dependencies on them. One could imagine a system which somehow \textit{generates} a random predicate $P$, in such a way that any version of an object can be classified as in $P$ or not, and then using that knowledge to generate dependency edges for predicate-based reads.

\textit{Conclusions}
\hspace{.3cm}
We present \Name{}: a novel theory and tool for experimental verification of transactional isolation. By using datatypes and generating histories which couple the version history of the database to client-observable reads and writes, we can extract rich dependency graphs between transactions. We can identify cycles in this graph, categorize them as various anomalies, and present users with concise, human-readable explanations as to why a particular set of transactions implies an anomaly has occurred.

\Name{} is sound. it identifies G0, G1a, G1b, G1c, G-single, and G2 anomalies, as well as inferring cycles involving per-process and real-time dependencies. In addition, it can identify dirty updates, garbage reads, duplicated writes, and internal consistency violations. When \Name{} identifies an anomaly in an observation of database, it must be present in every interpretation of that observation.

\Name{} is efficient. It is linear in the length of a history and effectively constant with respect to concurrency. It can quickly analyze real-world histories of hundreds of thousands of transactions, even when processes crash leading to high logical concurrency. We see no reason why it cannot handle more. It is dramatically faster than linearizability checkers~\cite{knossos} and constraint-solver serializability checkers~\cite{gretchen}.

\Name{} is effective. It has found anomalies in every database we've checked, ranging from internal inconsistency and aborted reads to anti-dependency cycles.

\Name{} is general. Unlike checkers which hard-code a particular example of an anomaly (e.g. long fork), \Name{} works with arbitrary patterns of writes and reads over different types of objects, so long as those objects and transactions satisfy some simple properties: traceability and recoverability. Generating random histories with these properties is straightforward; list append is broadly supported in transactional databases. \Name{} can also make limited inferences from less informative datatypes, such as registers, counters, and sets.

\Name{} is informative. Unlike solver-based checkers, \Name{}'s cycle-detection approach produces short witnesses of specific transactions. Moreover, it provides a human-readable explanation of \textit{why} each witness must be an instance of the claimed anomaly.

We are aware of no other checker which combines these properties. Using \Name{}, testers can write a small test which verifies a wealth of properties against almost any database. The anomalies \Name{} reports can rule out (or tentatively support) that database's claims for various isolation levels. Moreover, each witness points to particular transactions at particular times, which helps engineers investigate and fix bugs. We believe \Name{} will make the database industry safer.

\subsection{Acknowledgements}

The authors wish to thank Asha Karim for discussions leading to \Name{}, and Kit Patella for her assistance in building the \Name{} checker.

\newpage
\bibliographystyle{abbrv}
\bibliography{main}


\end{document}